\documentclass[aps, pra, twocolumn, superscriptaddress]{revtex4-2}
\usepackage[utf8]{inputenc}
\usepackage{amsmath, amsthm, amssymb}
\usepackage{amsthm}
\usepackage[utf8]{inputenc}
\usepackage[T1]{fontenc}
\usepackage{mathtools}
\usepackage[thinc]{esdiff}
\usepackage{amsfonts}
\usepackage{tikz}
\usepackage{braket}
\usepackage[english]{babel}
\usepackage{comment}
\usepackage{caption}
\usepackage{subfigure}
\usepackage{color}
\usepackage{graphicx} 
\usepackage[colorlinks=true,linkcolor=blue,urlcolor=blue,citecolor=blue]{hyperref}
\newtheorem{prop}{Proposition}

\definecolor{green(munsell)}{rgb}{0.0, 0.7, 0.47}

\begin{document}
\title{Non-Markovian noise limits for sustaining entanglement in multiparty quantum states}

\author{Suchetana Goswami}
\email{suchetana.goswami@gmail.com}

\affiliation{Harish-Chandra Research Institute, A CI of Homi Bhabha National Institute, Chhatnag Road, Jhunsi, Allahabad 211 019, India}

\affiliation{School of Informatics, University of Edinburgh, 10 Crichton St, Newington, Edinburgh EH8 9AB, UK}

\author{Ujjwal Sen}
\email{ujjwalsen0601@gmail.com}

\affiliation{Harish-Chandra Research Institute, A CI of Homi Bhabha National Institute, Chhatnag Road, Jhunsi, Allahabad 211 019, India}

\begin{abstract}
    
    The principal resource in several quantum information processing tasks 
    is  entanglement. 
    Building 
    practically useful versions
    of such tasks, one often needs to consider multipartite systems. 
    And in such scenarios, it is 
    impossible to 
    eliminate 
    the effects of environment while distributing the 
    resources. 
    We
    show how information back-flow from the environment, as a result of 
    non-Markovianity in the system dynamics, affects the resources in multipartite systems, specifically multiparty entanglement.
    In particular, we find that entanglement of multi-qubit Greenberger-Horne-Zeilinger cat states saturate with time to a non-zero value for a non-Markovian dephasing noise. The multi-qubit W states saturate to a higher value. Such non-zero time-saturated values are washed out if the non-Markovianity in the channel is removed. The multi-qubit W states provide a richer set of features, including a non-zero value for an asymptotic number of qubits and an even-odd dichotomy in entanglement with number of qubits. Moving over to a non-Markovian depolarising channel, we find that both cat and W states exhibit revival of entanglement after collapse. 
\end{abstract}

\maketitle

\section{Introduction}


The concept of entanglement remains perhaps the most important area of research in quantum information \cite{HHHH_09, ent2}, due to its crucial significance in applications in quantum 
technologies
as well as in fundamental aspects of quantum theory. Implementation of different information processing tasks or quantum computation in a distributive scenario \cite{EW_02, RP_15, ZZ_19, GH_23, KKI_99, MS_08} involve the presence of multipartite entanglement in
the physical substrate.
Despite its overwhelming importance, multipartite entanglement has been explored to a much lesser extent than the bipartite one. \\


Whenever we consider the physical realisation of any information processing task, a key factor to keep in mind is the interaction of the system with the environment. In the literature, there are different theoretical ways to handle or filter out  certain types of errors occurring due to the unavoidable presence of noise~\cite{PVK_97, KU_99, LB_03, SL_05, KJ_06, BSSL_07, BNPV_08, NSSFD_08, KLKK_12, DGPM_17, GGM_21},
but they are not always feasible in implementation. In general, the dynamics of such a system is governed by a non-unitary evolution, and with certain restrictions, can be described by the Gorini–Kossakowski–Sudarshan–Lindblad (GKSL) master equation~\cite{GKS_76, L_76}, which provides a completely positive and trace-preserving (CPTP) evolution of the system with time.
The dynamics is often assumed to be ``Markovian", and the typical assumptions involved include 
an uncorrelated system-environment initial state and 
an unaltered environment state whenever it interacts with the system~\cite{BP_02}. 
A typical dynamics of an open quantum system is often non-Markovian and that can be theoretically realised by removing the initial assumptions that are made in the Markovian case \cite{SL_09, BDMRR_13, LT_14, B_14, SL_16, SRS_19, CNB_22}.
%
There are different ways to detect and quantify non-Markovianity, and two popular ones are the Breuer-Laine-Piilo (BLP)~\cite{BLP_09} and the Rivas-Hulega-Plenio (RHP) measures~\cite{RHP_10}. The BLP measure 
is based on distinguishability of the evolved states, which basically arises from the concept of information back-flow from environment to the system~\cite{BLP_09}. On the other hand, 
the RHP measure is conceptualised by checking for 
non-monotonicity in the dynamics of system-auxiliary entanglement~\cite{RHP_10}.\\


In this paper, we consider non-Markovian versions of physically-motivated noisy quantum environments, and find their effects, when applied locally, on  multipartite entanglement.
%
%
We consider the effects on two paradigmatic  sets of multi-qubit quantum states, viz. the Greenberger-Horne Zeilinger (GHZ) ``cat'' states \cite{GHZ_89, M_90} and the W  states \cite{chabbis, DVC_00, SSWKZ_03}.
These states are local unitarily equivalent for two parties, but already for three  parties, they are  not inter-convertible via stochastic local quantum operations and classical communication (SLOCC)~\cite{DVC_00}
These states have appeared in numerous applications in quantum information protocols, and there have been a number of studies dealing with the entanglement content of these states in presence of noise \cite{SK_02, CMB_04, HDB_05, GBB_08, YE_09, KCP_14, MH_21, PWLXLW_23, XIG_23}.
In particular, the robustness of GHZ states in the asymptotic limit has been studied in presence of depolarising noise in~\cite{SK_02}.
The effects of noise on multipartite entanglement has largely been considered, hitherto, within the Markovian domain.\\


We show that the presence of non-Markovianity in the system
dynamics 
is potentially
helpful for sustaining entanglement of the state of the system. We notice that the effects of a particular noise on the two different 
sets 
of states are significantly 
different. 
We begin with the dephasing channel, with a non-Markovian element,
and show that the multi-qubit W state is more robust than the multi-qubit GHZ in sustaining their entanglement contents. 
We find a dichotomy between even and odd number of qubits in case of multipartite W states when a particular bipartition is considered. 
Next, we consider depolarising noise, and show that it is possible to retrieve the entanglement of the output state by introducing a memory effect in the dynamics. This is true for both W and GHZ states.\\

The paper is organised in the following way. In Section \ref{sec_2}, we introduce the initial systems and the noise models we will be working with. We also portray the exact structure of the protocol. Next, in Section \ref{sec_3}, we discuss 
in some detail 
about the properties of the final state entanglement in presence of non-Markovian versions of the considered channels. We present a conclusion in Section~\ref{sec_4}.

\section{The initial system and the decoherence models}
\label{sec_2}

In this paper, we consider some noise models that may arise in realistic situations and we characterise their effects on multipartite systems, where each party is assumed to be holding one qubit. The initial state of the system is considered to be either a GHZ or a W state. We also assume that the system and the environment are completely uncorrelated in the beginning of the dynamics. Mathematically the joint initial state of the system and environment is hence considered to be in a product form. For any $N$-qubit system, the GHZ state 
has the following form,
\begin{equation}
    \ket{\psi}_{GHZ}^N = \frac{1}{\sqrt{2}} (\ket{0}^{\otimes N} + \ket{1}^{\otimes N}).
    \label{ghz_state}
\end{equation}
On the other hand, the $N$-qubit W state 
can be written as
\begin{eqnarray}
    \ket{\psi}_{W}^N &=& \frac{1}{\sqrt{N}} (\ket{0}^{\otimes (N-1)} \ket{1} + \ket{0}^{\otimes (N-2)} \ket{1} \ket{0}+ \nonumber\\
    && \cdots + \ket{0}\ket{1} \ket{0}^{\otimes (N-2)}
    + \ket{1} \ket{0}^{\otimes (N-1)}).
    \label{w_state}
\end{eqnarray}
The state in Eq. (\ref{w_state}) is a single-mode excited W  state. We also consider multi-mode excited
W states. Note that the GHZ and the single-mode excited W states are 
SLOCC-inequivalent. In \cite{SK_02}, the authors have dealt with the robustness of GHZ states in presence of depolarising noise while the number of parties increases asymptotically. They considered the form of depolarising noise where the initial state is modified by the presence of local white noise, and the channel acting on every party transforms the input of that party into a probabilistic mixture in the following way: $\rho \rightarrow (1-p) \rho + p~\mathcal{I}/d$, where $\mathcal{I}$ is the corresponding identity matrix, \(p\) is the probability of the noise admixed, and \(d\) is the corresponding local Hilbert space dimension. While this is an interesting way for describing noise, 
it is a somewhat restrictive description of the interaction (of the system with its environments), as it does not explicitly describe evolution of an initial state with time. The dynamics of the state of the system with time under such noise can only be evaluated via solving the master equation governing the evolution of the state under such noise, with the exact form of the master equation being derived from the interaction Hamiltonian of the system with the environment.

\begin{figure}[htp]
    \fbox{\includegraphics[scale=0.12]{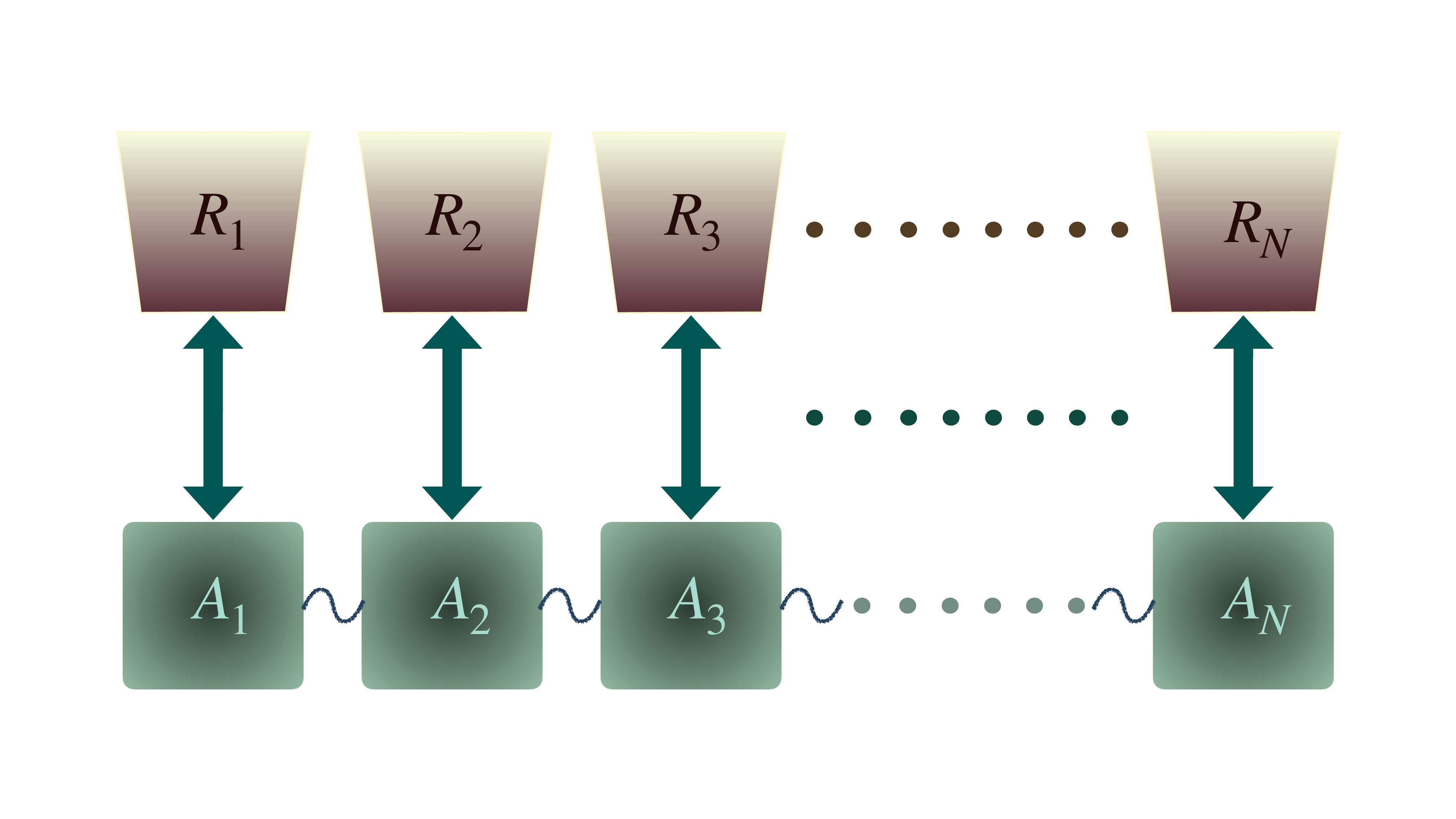}}
    \caption{\footnotesize{Schematic diagram of system-environment interaction. Here $N$ parties, $A_1$ to $A_N$ each holding a single qubit, share an entangled state and each qubit of the system is interacting with $N$ identical non-interacting reservoirs, $R_1$ to $R_N$.}}
    \label{int_NM}
    \end{figure}

In our prescription, we have an $N$-qubit system initialised in either of the states given in Eqs.~(\ref{ghz_state}) and~(\ref{w_state}). These $N$ qubits are considered to be interacting with $N$ reservoirs individually.
We consider that all the reservoirs describing the environment, are identical and the interactions between any qubit and its environment are also considered to be similar. We illustrate the scenario in Fig.~\ref{int_NM}. Hence our protocol corresponds to a model where the system goes through local interaction with a noisy environment. 
The system-environment cluster evolves for a time \(\tilde{t}\), 
and we are interested in the entanglement characterisation of the final state.
For a single qubit, the evolution of the state with time is governed by the time-local master equation,
\begin{eqnarray}
&&   \diff{\rho}{(\omega_0 \tilde{t})} = \dot{\rho}(t) = -\frac{i}{\hbar \omega_0} [H(t),\rho(t)] \nonumber\\
    && +\frac{1}{\omega_0}\sum_k \gamma_k (t) \left( L_k (t) \rho(t) L_k^{\dagger}(t)  -\frac{1}{2} \{L_k^{\dagger}(t) L_k(t) , \rho(t) \} \right).\nonumber\\
\end{eqnarray}
We call, $\omega_0 \tilde{t}=t$. Here, $H(t)$ is the Hamiltonian of the system under consideration, $L_k(t)$ are the Lindblad operators corresponding to the $k$-th bath, and $\gamma_k(t)$ are the corresponding rates of dissipation. The first part of the above equation represents the unitary part of the dynamics, 
while the next part governs the dissipative scenario (interaction picture). The time dependence of the decay rates introduces  memory effects in the system-bath interaction which in-turn brings in  non-Markovian features in the dynamics of the \(N\)-qubit system.

Before moving on to the explicit description of the noise models, here we discuss about some particular features of the initial states considered. Note that, the $N$-qubit states considered are symmetric across any $(N-m)~vs.~m$ bipartition, where any $m$ parties among $N$ (holding an $m$-qubit system) are standing together and the rest $(N-m)$ parties are also standing together. This is how the bipartition is created. For even number of qubits, $1\leq m \leq N/2$ and for odd number of qubits, $1\leq m \leq (N-1)/2$. In the following we give an example for better understanding. Let us consider a five-qubit system, i.e. $N=5$ distributed among five parties $A_1$, $A_2$, $A_3$, $A_4$, and $A_5$. In this case, there can be two unique values for $m$ and they are $m=1$ and $m=2$.
\begin{prop}
    For a five-qubit system, the GHZ and the W states are symmetric across all possible unique bipartitions.
\end{prop}
\begin{proof}
    Let us first consider $m=1$. In this case, the bipartitions possible are, $A_1~vs.~A_2A_3A_4A_5$, $A_2~vs.~A_1A_3A_4A_5$, $A_3~vs.~A_1A_2A_4A_5$, $A_4~vs.~A_1A_2A_3A_5$, and $A_5~vs.~A_1A_2A_3A_4$. In all these bipartitions the GHZ state can be written as, $(\ket{0}\ket{0_4}+\ket{1}\ket{15_4})/\sqrt{2}$. On the other hand the W state can be written as $(\ket{0}\ket{1_4}+\ket{0}\ket{2_4}+\ket{0}\ket{4_4}+\ket{0}\ket{8_4}+\ket{1}\ket{0_4})/\sqrt{5}$. Here the basis for the four-qubit system is considered as $\ket{i_5}_{i=0}^{15}$. Similarly for $m=2$, the bipartitions are $A_1A_2~vs.~A_3A_4A_5$ and all possible such permutations with two parties standing together and the rest three parties are also standing together. In this situation, the GHZ state can be written as $(\ket{0_2}\ket{0_3}+\ket{3_2}\ket{7_3})/\sqrt{2}$, and the W state can be re-written as, $(\ket{0_2}\ket{1_3}+\ket{0_2}\ket{2_3}+\ket{0_2}\ket{4_3}+\ket{1_2}\ket{0_3}+\ket{2_2}\ket{0_3})/\sqrt{5}$. Here the basis for the two-qubit and the three-qubit systems are ${\ket{j_2}_{j=0}^3}$ and ${\ket{k_3}_{j=0}^7}$ respectively. Hence both the states are symmetric under the particular bipartition considered.
\end{proof}
The above claim can be extended to any $N$. The noise models considered in such a way that they affect each qubit in the system in the same way. Hence the interaction with the environment does not change the initial symmetry of the considered states. Hence for the entanglement, we calculate the logarithmic negativity of the final states in all possible bipartitions. Interestingly enough, we see that under dephasing noise, while for GHZ states the values of logarithmic entanglement are same for all bipartitions, for W state, this is not the case. This is due to the fact that, for GHZ states, when written in any bipartition, they ultimately resemble to the structure of a bipartite maximally entangled state. On the other hand, in presence of depolarising noise, GHZ states do not retain this feature of equivalence across all possible bipartitions. Intuitively, this can be an effect of the fact that the depolarising noise takes into account the dimension of the system, which makes all bipartitions behave differently in presence of the noise, even for the GHZ states.\\

\noindent \textit{Local dephasing noise.} Let us start by considering that the qubits are locally interacting with some bosonic reservoirs and hence the total system-reservoir Hamiltonian for a single qubit is given by the sum of the system Hamiltonian ($H_S$), the Hamiltonian of the reservoir ($H_B$), and the interaction Hamiltonian ($H_I$). The particular form that we assume is given by,
\begin{equation}
    H_D/\hbar= \omega_0 \sigma_z + \sum_i \omega_i a_{i}^{\dagger} a_i + \sum_i \sigma_z (g_i a_i + g_i^* a_i^{\dagger}),
    \label{dephase_H}
\end{equation}
where $2\hbar\omega_0$ is the energy gap of the qubit energy levels, $\omega_i$'s are the frequencies of the reservoir modes, $a_i$ ($a_i^{\dagger}$) is the annihilation (creation) operator of the $i$th mode of the reservoir, and $g_i$ is the coupling constant between the system and the \(i\)th mode of the reservoir. Here, $\sigma_z$ is the Pauli z matrix, with the local computational basis being assumed to the eigenstates of \(\sigma_z\). In the interaction picture, the exact master equation corresponding to this Hamiltonian takes the simple form,
\begin{equation}
    \dot{\rho}(t)= \frac{\gamma(t)}{\omega_0} \left( \sigma_z \rho(t) \sigma_z - \rho(t) \right).
    \label{dephase_ME}
\end{equation}
 We assume that the initial state of the total system is in product form in the qubit to environment partition, and the environment is assumed 
 to be initialised in the thermal state at a temperature $T$, so that the initial state of the environment is $e^{-\beta H_B}/\text{Tr}[e^{-\beta H_B}]$ with $\beta=1/(k_B T)$), with \(k_B\) being the Boltzmann constant. Following these assumptions the time-dependent dephasing rate takes the form \cite{BP_02}, 
\begin{equation}
    \gamma(t) = \int d\omega~\mathcal{S}(\omega)~\text{coth}\left(\frac{\hslash \omega}{2 k_B T}\right) \frac{\text{sin}(\omega t)}{\omega}.
    \label{int_gamma}
\end{equation}
Here, $\mathcal{S}(\omega)$ is the reservoir spectral density, which, in the continuum limit of the modal frequencies, is related to $g_i$'s as $\sum_i |g_i|^2 \rightarrow \int d\omega \mathcal{S}(\omega) \delta(\omega_i -\omega)$. The dynamics given in Eq.~(\ref{dephase_ME}) represents the usual dephasing channel having the usual Kraus operators $K_0 = \sqrt{1-\eta(t)/2}~\mathcal{I}$ and $K_1 = \sqrt{\eta(t)/2}~\sigma_z$. The time-dependent dephasing factor $\eta(t)$ is related to $\gamma(t)$ as $\eta(t)=1-\text{exp}[-2\int_0^t \gamma(t') dt']$. Now let us consider that the reservoir is governed by the Ohmic spectral density function given by 
\begin{equation}
    \mathcal{S}(\omega) = \frac{\omega^s}{\omega_c^{s-1}} e^{-\omega/\omega_c},
\end{equation}
where $\omega_c$ is the cut-off frequency of the reservoir and $s$ is the Ohmicity parameter. Using this spectral density in Eq.~(\ref{int_gamma}), one can get the exact expression for the dephasing decay rate at a given temperature~\cite{HJM_13}. At zero temperature it takes the form,
\begin{equation}
    \gamma_{T=0}(t,s)= \omega_c [1+(\omega_c t)^2]^{-\frac{s}{2}} \Gamma[s] \text{sin}(s~\text{arctan} (\omega_c t)).
    \label{gamma_dephase}
\end{equation}
Here $\Gamma[\cdot]$ denotes the Euler gamma function. Note that 
for $s > 2$, the decay rate takes negative values for a given time range. Such negative values of decay rate can be argued to bring in   memory effects in the system's  dynamics. This in turn induces the features of non-Markovianity in the evolution as there is a possibility of back-flow of information from the reservoir to the attached qubit. 

In this paper, we consider the cases with different values of $s$ and show how the entanglement of the final state saturates to a non-zero value with respect to time. We show that with the increasing number of qubits, though the value at which the final state entanglement saturates decreases, the decrement is converging with the number of qubits.\\

\noindent \textit{Local depolarising noise.} Let us now consider a further system-environment coupling  where the total Hamiltonian of the system and the bath is given by
\begin{equation}
    H_{DP}/\hbar= \omega_0 \sigma_z + \sum_i \omega_i a_{i}^{\dagger} a_i + \sum_{j=x,y,z} \sum_i \sigma_j (g_i a_i + g_i^* a_i^{\dagger}).
    \label{depo_H}
\end{equation}
Considering this Hamiltonian, the exact master equation takes the following form:
\begin{equation}
    \dot{\rho}(t)= \sum_{j=x,y,z} \frac{\gamma_j(t)}{\omega_0} \left( \sigma_j \rho(t) \sigma_j - \rho(t) \right).
    \label{depo_ME}
\end{equation}
Here, $\sigma_x$, $\sigma_y$ and $\sigma_z$ are Pauli matrices. Note that, for $\gamma_x (t) = \gamma_y (t) = 0$, we get back the case corresponding to the dephasing channel. Now, in the previous case, the explicit form of the decay rate can be evaluated because the system Hamiltonian commutes with the total Hamiltonian $H_D$ given in Eq. (\ref{dephase_H}). Finding the explicit forms of the $\gamma_j$'s become much involved in this case as the interaction Hamiltonian has changed. Instead, here we consider the forms of the decay rates keeping certain constraints in mind. From \cite{CW_13, CM_14}, in case of above evolution, the dynamics is Markovian (i.e. complete positivity (CP)-divisible) if and only if $\gamma_x(t) \geq 0$, $\gamma_y(t) \geq 0$, and $\gamma_z(t) \geq 0$ for all time $t \geq 0$. Alongside, only P-divisibility is equivalent to the following conditions:
\begin{align}
    \gamma_x(t)+\gamma_y(t) \geq 0,\nonumber\\
    \gamma_y(t)+\gamma_z(t) \geq 0,\nonumber\\
    \gamma_z(t)+\gamma_x(t) \geq 0,
    \label{p-div_cond}
\end{align}
for all time $t \geq 0$. Violation of any of these inequalities in Eq. (\ref{p-div_cond}) implies that the dynamics is essentially non-Markovian. Note that, breaking P divisibility definitely suggests CP divisibility but the reverse is not true. This follows from the concepts of $k$-positive maps. 

In this paper, we choose one of the $\gamma$'s take negative values for certain period of time $t$. Explicitly, we take $\gamma_z(t) = \alpha~\text{sin} (t)$ which definitely makes $\gamma_z(t) \leq 0$ for certain range of $t$ breaking the CP divisibility. Evidently, the P divisibility would be broken depending on the choices of values of the other two decay rates.

\section{Characterising the final state entanglement under noise}
\label{sec_3}

\noindent \emph{Results for Dephasing noise.} Let us start by describing the simplest case involving three qubits, i.e. $N=3$ distributed among three parties $A_1$, $A_2$, and $A_3$. First, let us assume that the system is initialised in GHZ state. Hence the initial state looks like $\ket{\psi}_{GHZ}^3 = (\ket{000}+\ket{111})/\sqrt{2}$. According to our protocol, each qubit interacts with identical baths (three baths in this case) at zero temperature at any instant of time $t$. 
\begin{figure}[htp]
\centering
\fbox{
\subfigure[$3$-qubit GHZ state (N=3)]{\includegraphics[scale=0.1]{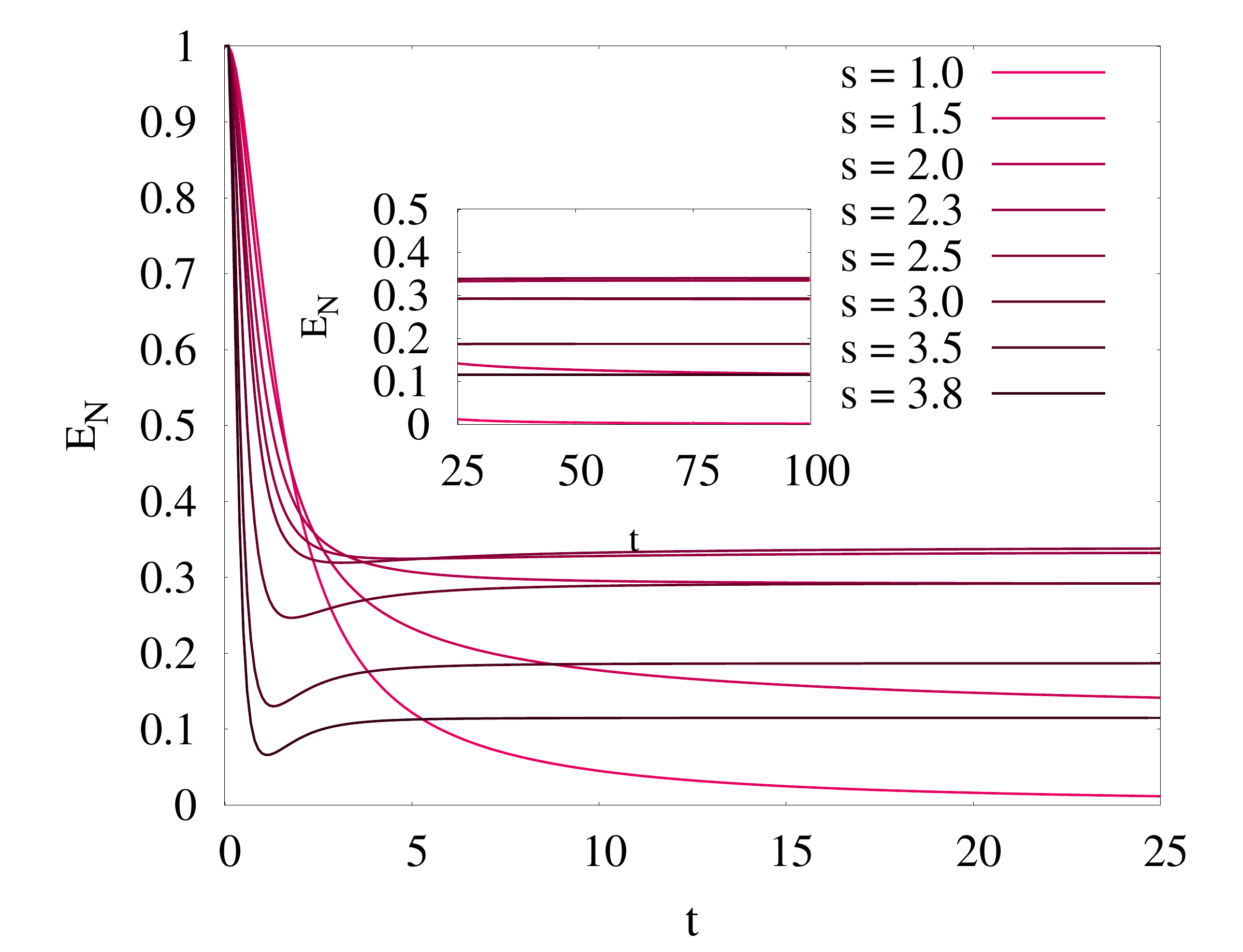}}
\quad
\subfigure[$5$-qubit GHZ state (N=5)]{\includegraphics[scale=0.1]{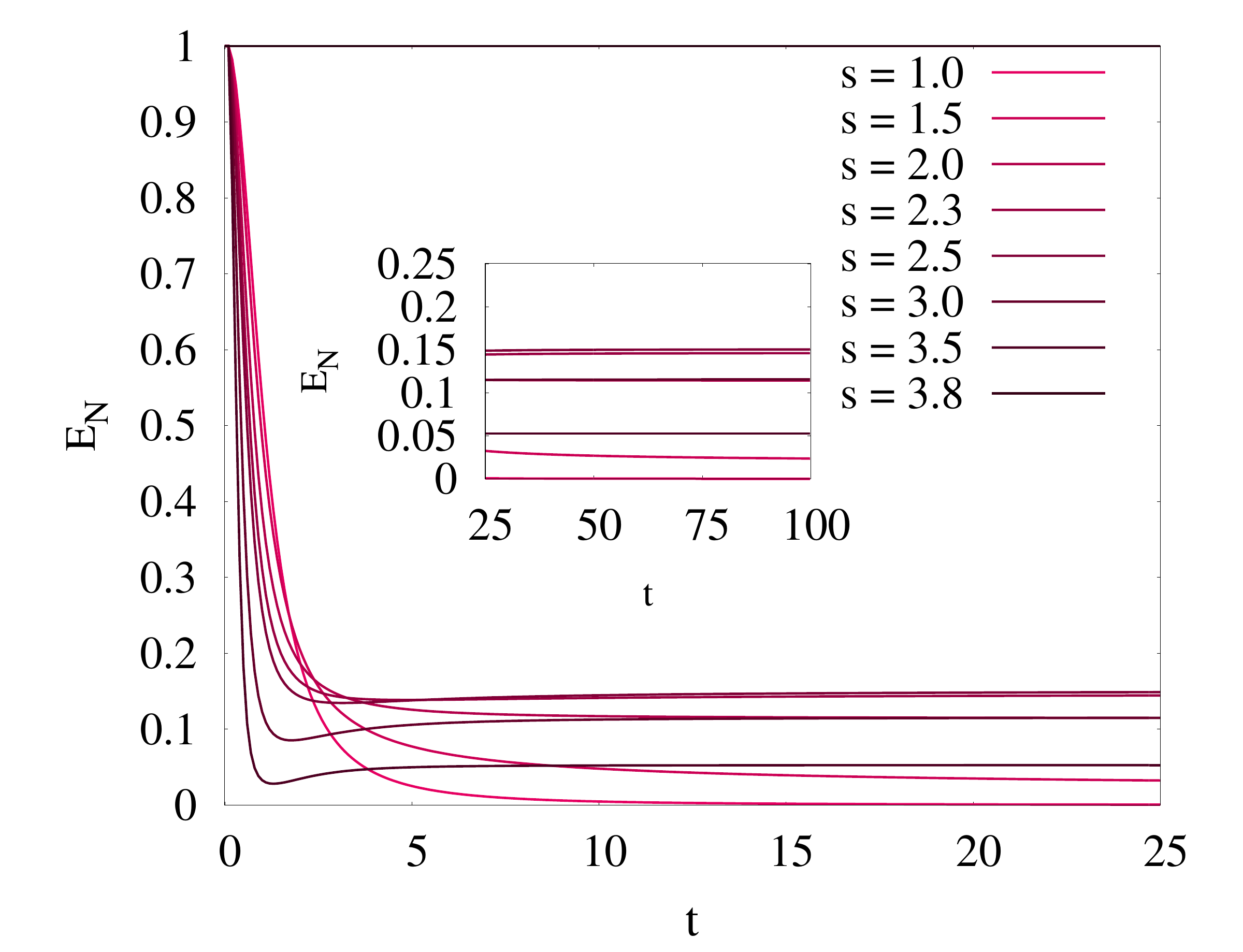}}}
\caption{\footnotesize{Entanglement of the final state under dephasing noise with time when the input state is a GHZ state. The vertical axes are in ebits, while the horizontal ones are dimensionless. 
}}
\label{E_t_GHZ}
\end{figure}

\begin{figure*}[htp]
\centering
\fbox{
\subfigure[Entanglement variation with decay rate for different number of qubits at $t=30$.]{\includegraphics[scale=0.135]{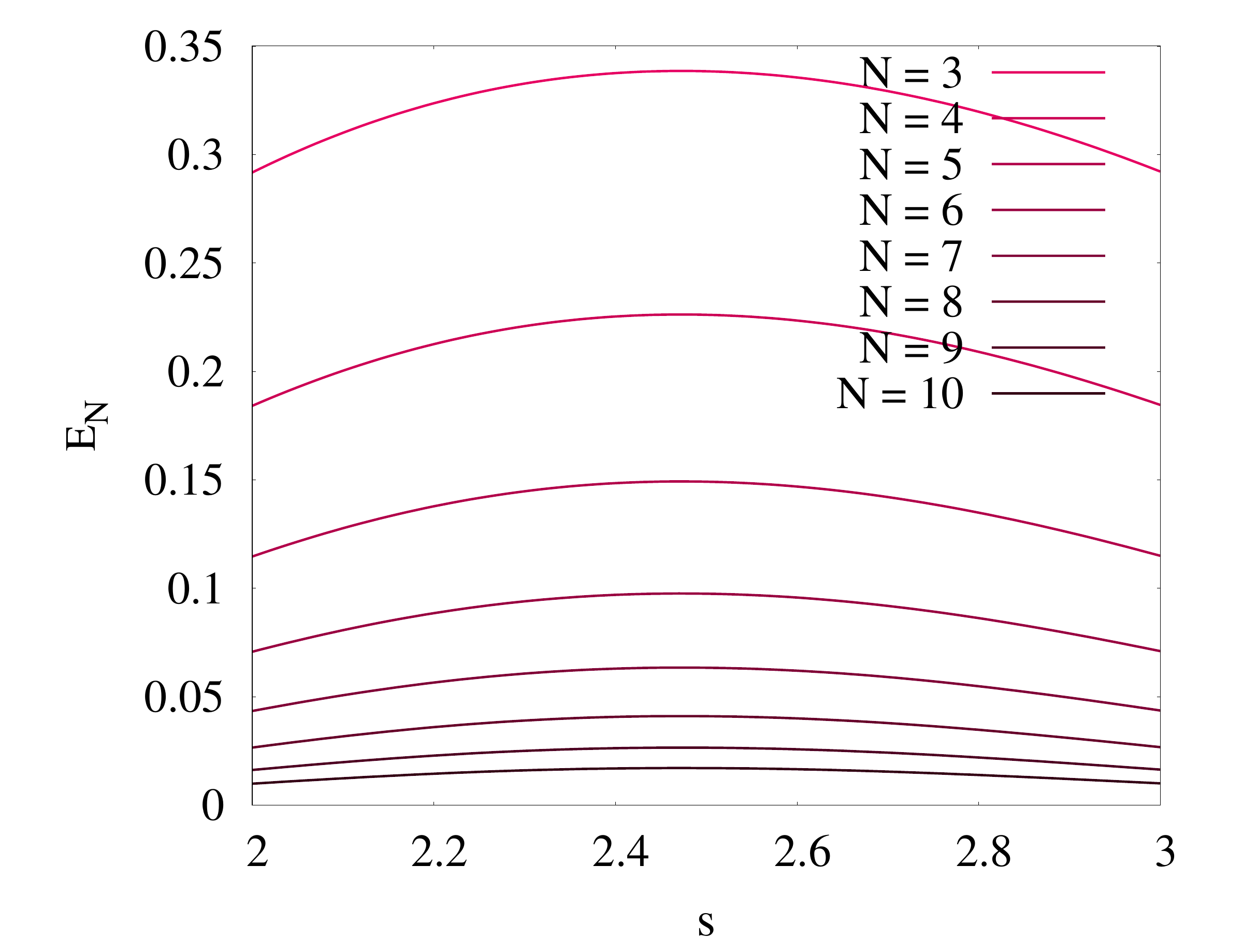}}
\quad
\subfigure[Decay of entanglement with number of qubits for different values of decay rate at $t=30$.]{\includegraphics[scale=0.135]{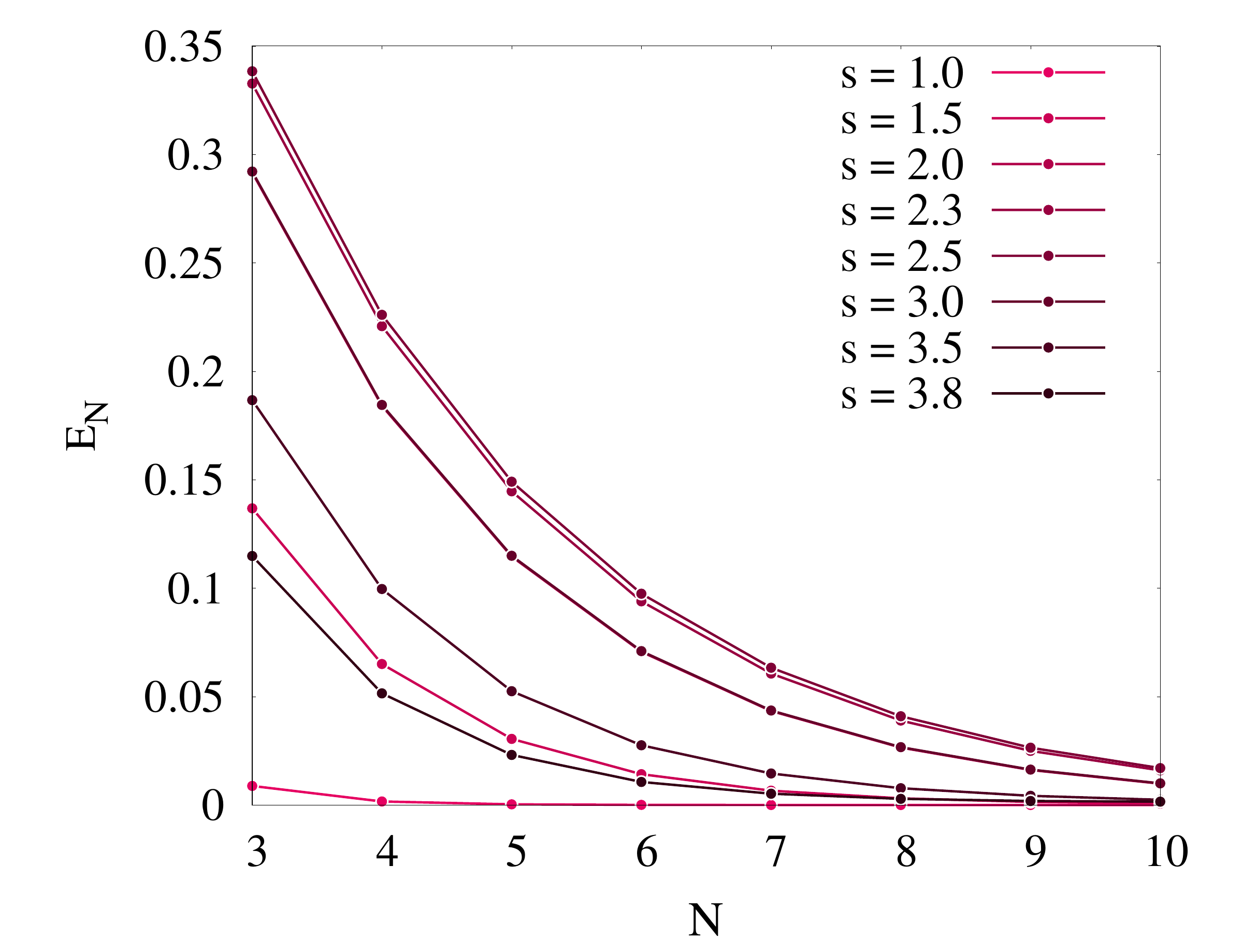}}
\quad
\subfigure[Decay of entanglement with number of qubits for $s=2.47$ at $t=30$. Here the dots indicate the exact data points, while $f(N)$ gives the extrapolated version of the same.]{\includegraphics[scale=0.135]{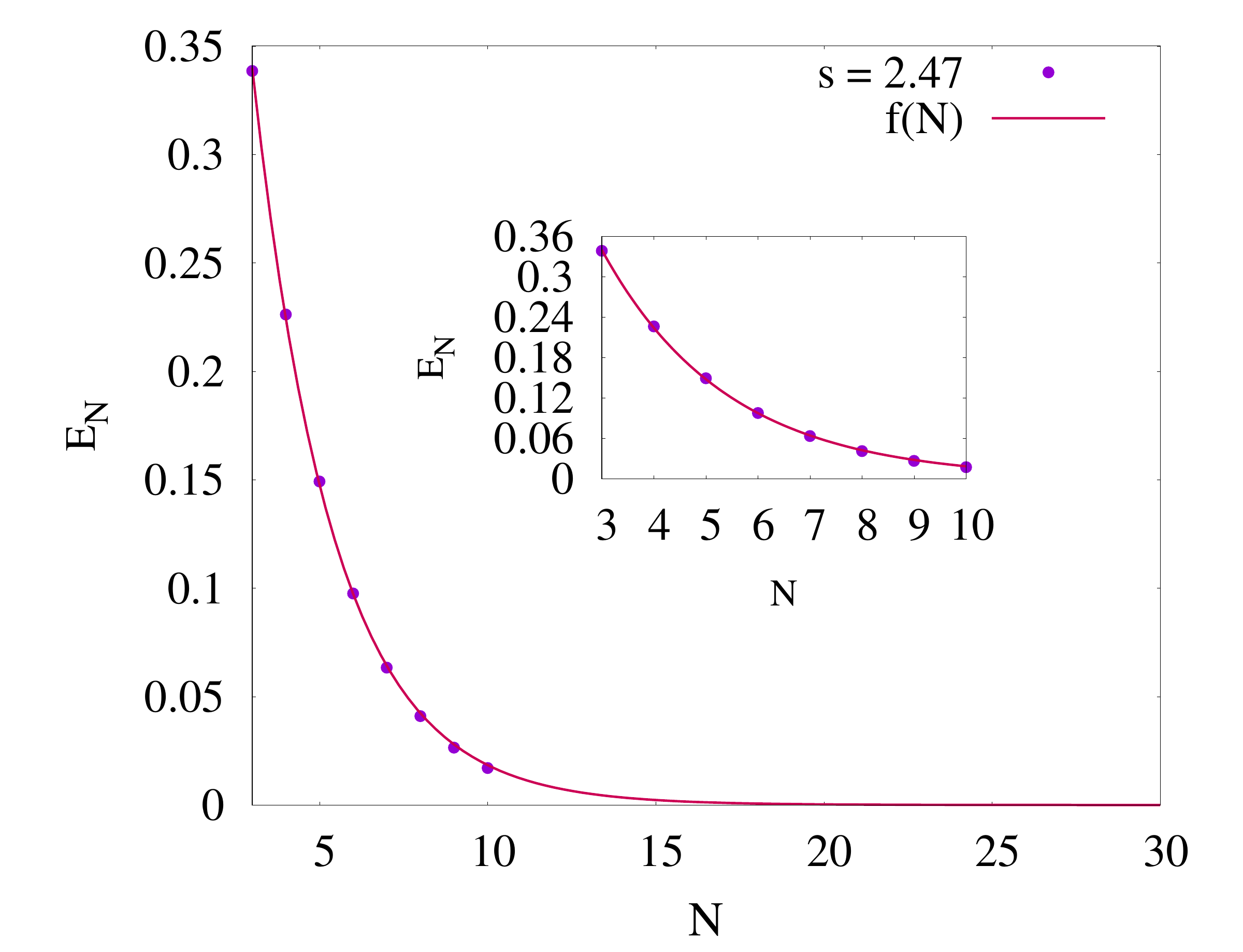}}}
\caption{\footnotesize{Entanglement characterisation of the final state generated by applying dephasing noise on GHZ state. The vertical axes are in ebits, while the horizontal ones are dimensionless. }}
\label{E_N_s_t30_GHZ}
\end{figure*}

To obtain the evolution of the state of the system undergoing the environmental interaction, we solve the master equation. In this particular case, mathematically one can write the dynamics of the state as,
\begin{eqnarray}
    \dot{\rho}(t) &=& \frac{\gamma(t)}{4 \omega_0} [ (\sigma_z \otimes \mathbb{I} \otimes \mathbb{I}) \rho(t) (\sigma_z \otimes \mathbb{I} \otimes \mathbb{I}) \nonumber\\
    && + (\mathbb{I} \otimes \sigma_z \otimes \mathbb{I}) \rho(t) (\mathbb{I} \otimes \sigma_z \otimes \mathbb{I}) \nonumber\\
    && + (\mathbb{I} \otimes \mathbb{I} \otimes \sigma_z) \rho(t) (\mathbb{I} \otimes \mathbb{I} \otimes \sigma_z) - \rho(t) ]
    \label{3qubit_dephase}
\end{eqnarray}
In the above equation the dynamics is considered to be governed by the dephasing noise as can be seen from the dissipative part of the equation. In case of depolarising noise, Eq. (\ref{3qubit_dephase}) is accordingly modified. As the number of qubits increases, the above expression will be modified in an iterative fashion. Now, we let the system evolve from $t=0$ to $t=100$ under the given dynamics, where the decay rate $\gamma(t)$ is given by Eq. (\ref{gamma_dephase}). At every instant of time, we find out the state $\rho(t)$ by solving the Eq. (\ref{3qubit_dephase}). For the particular case of three-qubit system we evaluate the entanglement of the final state in any of the $one~vs.two$ bipartitions, i.e. either in $A_1~vs.~A_2A_3$, or in $A_2~vs.~A_3A_1$, or in $A_3~vs.~A_1A_2$, as the state remains symmetric across all such bipartitions. We evaluate logarithmic negativity of the given bipartite system considering the above mentioned bipartitions. 

In Fig.~\ref{E_t_GHZ}, the variation of logarithmic negativity of entanglement of the final state in any given bipartition is plotted against time. Note that, for different values of $s$, we obtain different saturation for the final state entanglement. For $s<2$ the dynamics is surely Markovian and in that domain the value of logarithmic negativity reduces with time and touches zero after a certain period of time. However, the non-Markovian features in the dynamics for $s>2$ induce information back-flow from the environment to the system, which in turn helps to retain a non-zero value of entanglement in the output end and the final value saturates with time. In the very beginning of the time scale, the entanglement falls down and retrieves the value where it finally saturates. As the number of qubits increases, the saturation of entanglement occurs at a lower value for any bipartition. Two particular cases having two different number of qubits i.e. for $N=3$ and $N=5$, the characterisations of the final state entanglement with time are illustrated in Fig (\ref{E_t_GHZ}(a) and (b)) respectively. Numerically, it is evident that the similar nature of the entanglement of the output state can be observed for $N=3$ to $N=10$ in the non-Markovian domain of $s$ for a large period of time.


Now, we concentrate only on the region where the decay rate takes negative value i.e. for $s>2$. For these values of $s$, the final state retains the entanglement more than the Markovian domain. From Fig. (\ref{E_N_s_t30_GHZ} (a)), it is evident that the final state retains the maximum amount of entanglement in the region between $s=2.3$ to $s=2.5$ for all $N$. Interestingly we find that the final state takes the maximum saturating value of entanglement for $s=2.47$ and this is valid for any $N$. Fig. (\ref{E_N_s_t30_GHZ} (b)) indicates how the saturating value of final state entanglement falls off with increasing number of qubits. From these plots, it can also be noticed that the maximum improvement occurs between the same range of $s$. Now we plot the particular case corresponding to $s=2.47$ in Fig. (\ref{E_N_s_t30_GHZ} (c)). Here, the value of entanglement falls with the number of qubits as $E_N = a e^{-c (N-3)} +b^2$ with $a=0.3399$, $b=7.0847 \times 10^{-5}$ and $c=0.4167$. Evidently, for other values of $s$ the decay of entanglement with increasing number of qubits follow the same pattern with different parameterisation of $a$, $b$ and $c$. 
From the extrapolation of $E_N$ with respect to $N$ as illustrated in Fig (\ref{E_N_s_t30_GHZ} (c)) after $N=10$, we see that the final state retains a non-zero value of entanglement up to $N=16$ in case of dephasing noise before it gets completely diminished by the presence of noise.\\

\begin{figure*}[htp]
\centering
\fbox{
\subfigure[Decay of entanglement in $1$-Rest with number of qubits for different values of decay rate at t = 30.]{\includegraphics[scale=0.1]{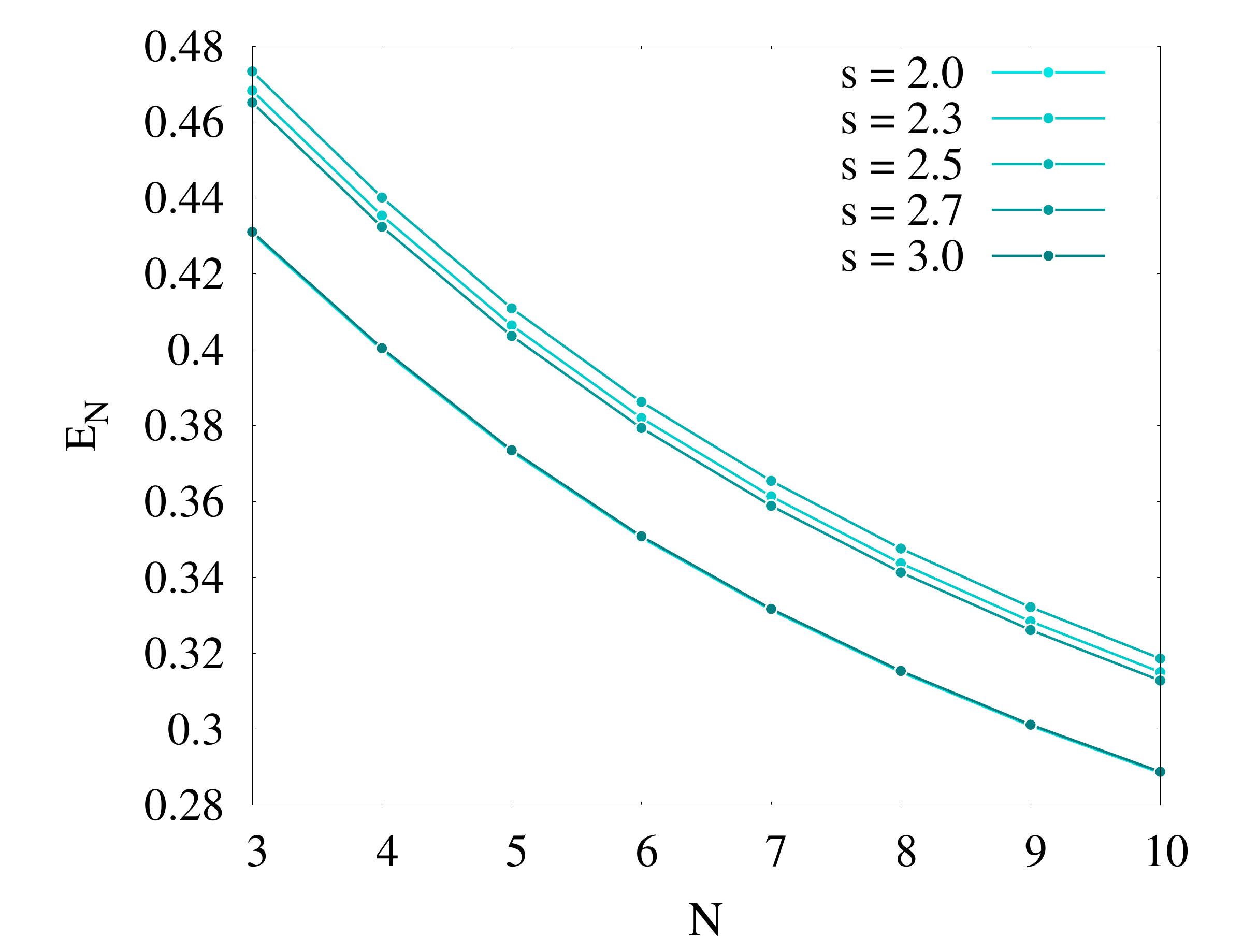}}
\quad
\subfigure[Entanglement pattern in highest-cut with number of qubits for different values of decay rate at t = 30.]{\includegraphics[scale=0.1]{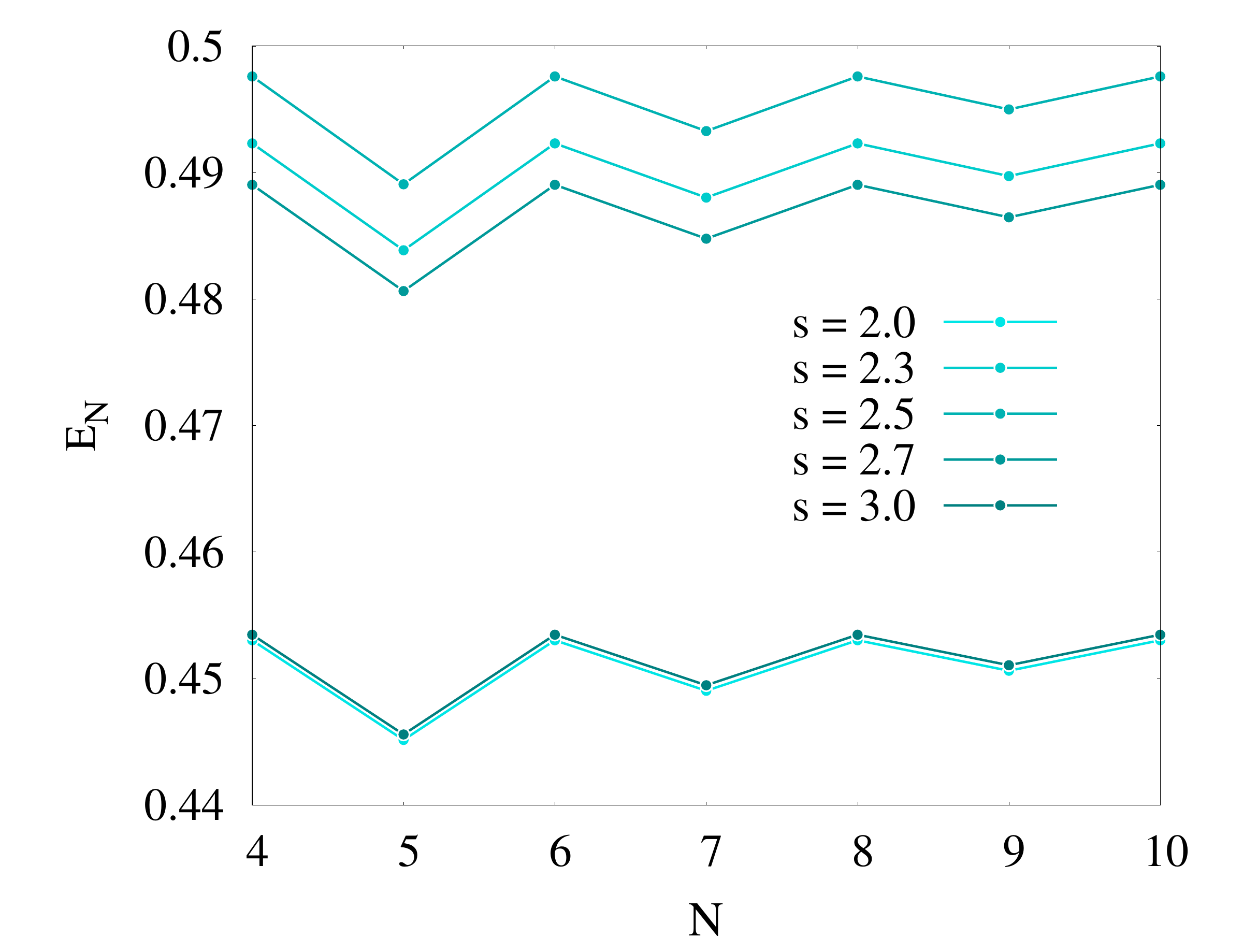}}
\quad
\subfigure[Decay of entanglement in $1$-Rest with number of qubits for $s=2.47$ at $t=30$.]{\includegraphics[scale=0.1]{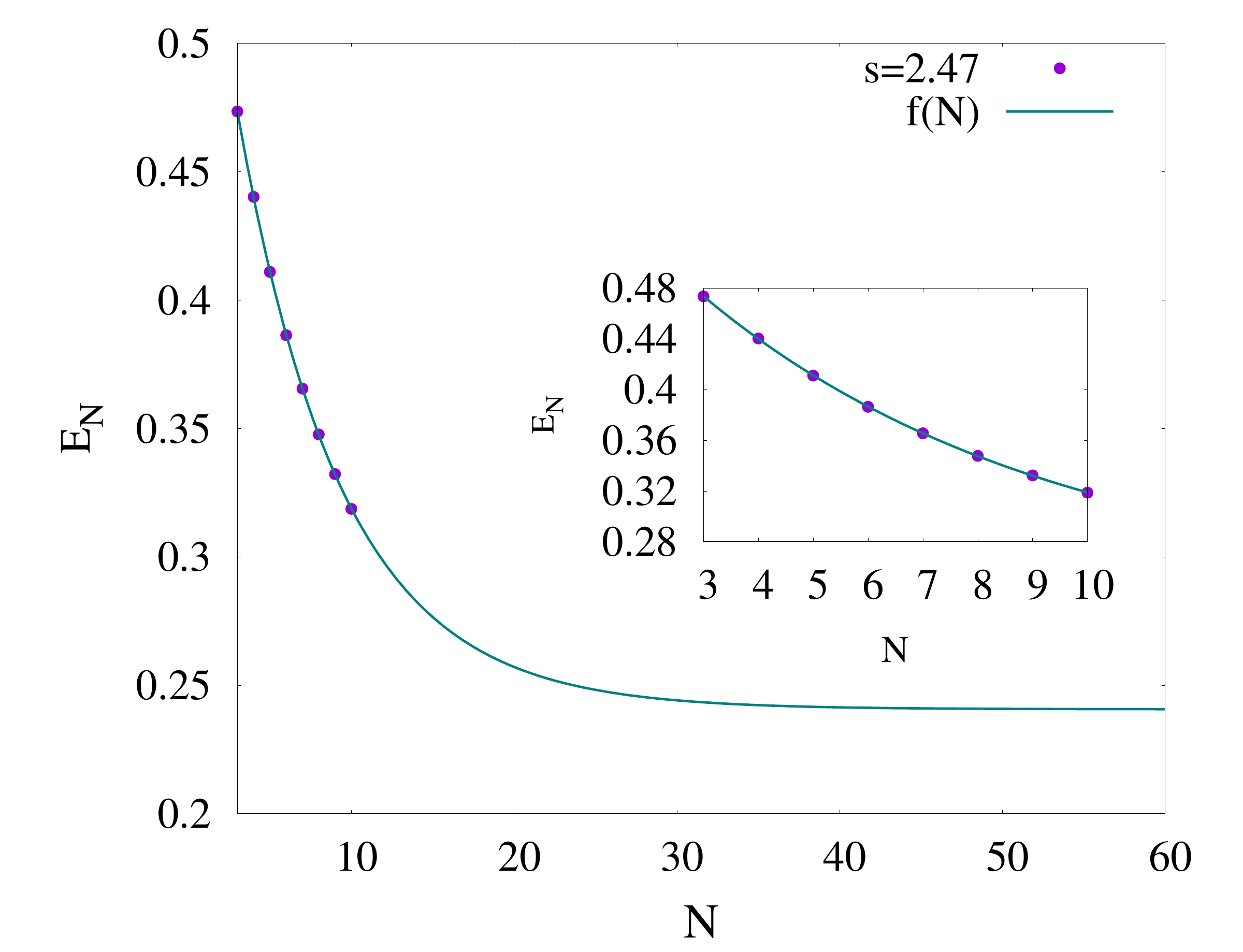}}
\quad
\subfigure[Entanglement pattern in highest-cut with number of qubits for $s=2.47$ at $t=30$.]{\includegraphics[scale=0.1]{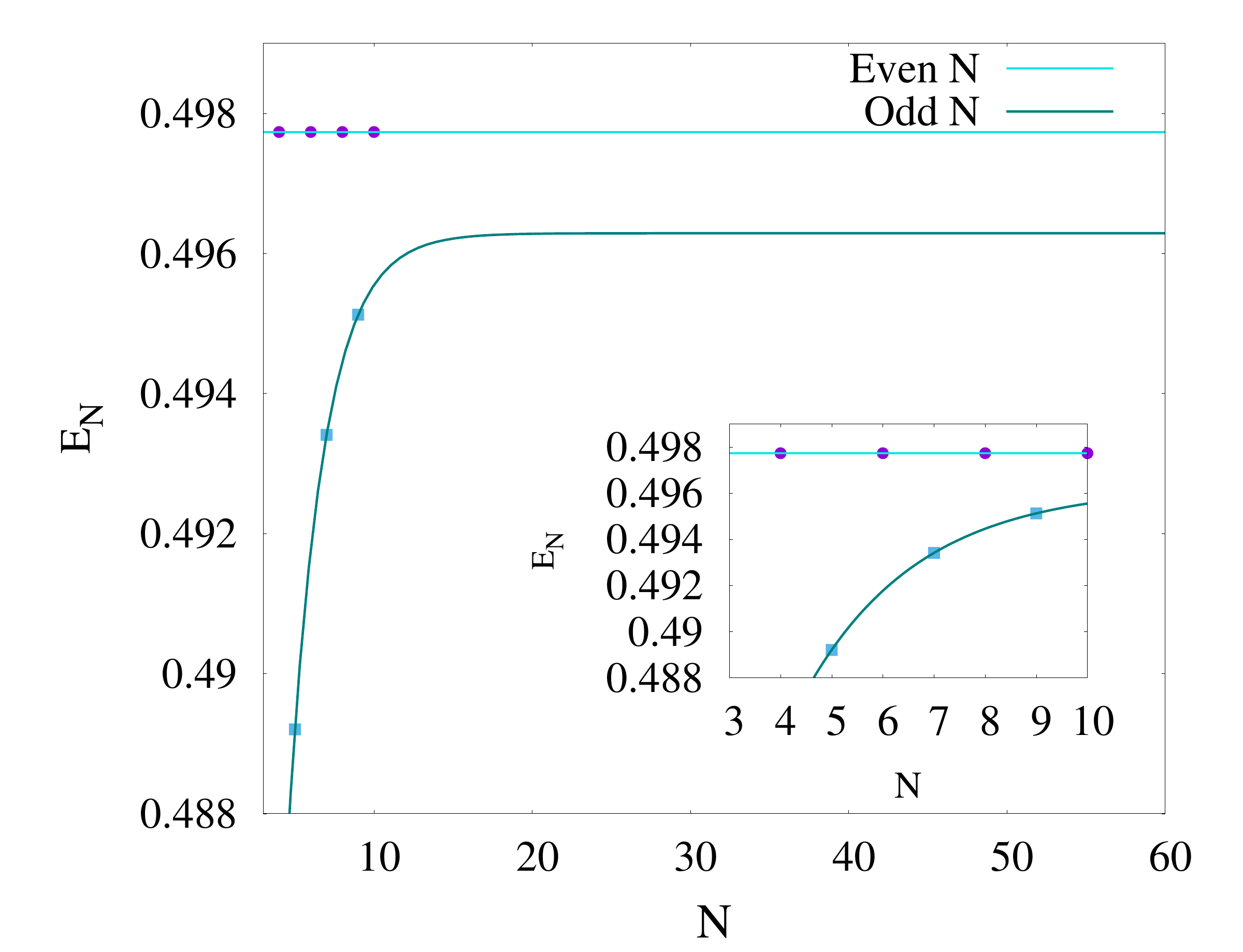}}}
\caption{\footnotesize{Entanglement characterisation of the final state generated by applying dephasing noise on W state. The vertical axes are in ebits, while the horizontal ones are dimensionless. }}
\label{E_N_s_t30_W}
\end{figure*}


Next, we move our discussion to the case when the initial state is a W states with single to $(N-1)$ mode excitation. The prescription for the interaction with the environment is same as the previous case and the dynamics of the state under the noise is given in Eq. (\ref{3qubit_dephase}). In this case the results are more involved as all the bipartitions are not equivalent except for $N=3$. From $N=4$ to $N=10$, we explicitly discuss two cases, i) the ``$1$-Rest" scenario, when the bipartition is between \textit{one} party standing alone with the rest of the parties standing together, and ii) the ``highest-cut" scenario. In case of the ``highest-cut", for even $N$ the bipartition is between $N/2$ parties with the remaining $N/2$, whereas for odd $N$, the bipartition is between the $(N-1)/2$ parties with other $1+(N-1)/2$ parties. For example, when $N=7$ and the system is distributed among $A_1$, $A_2$, $A_3$, $A_4$, $A_5$, $A_6$, $A_7$, then the unique bipartitions possible are, $A_1~vs.~A_2A_3A_4A_5A_6A_7$, $A_1A_2~vs.~A_3A_4A_5A_6A_7$, $A_1A_2A_3~vs.~A_4A_5A_6A_7$. Any other bipartition possible, are mere permutation of the parties, under which the state remains symmetric, as discussed earlier. Hence for $N=7$, we call the $A_1~vs.~A_2A_3A_4A_5A_6A_7$ bipartition as the ``$1$-Rest" case and $A_1A_2A_3~vs.~A_4A_5A_6A_7$ bipartition as the ``highest-cut" case. The similar nomenclature is extended for any $N$.
We let the system evolve under the dephasing noise from $t=0$ to $t=100$ and analyse the characterisation of entanglement during the evolution and after a certain time passed. As we know that the non-Markovian effects are only prominent for $s>2$ in Eq. (\ref{gamma_dephase}), here we concentrate on the particular region. We see that similar to the case of GHZ states, the value of output state entanglement (in terms of logarithmic negativity) first decreases and then saturates at a particular value with increasing time for all the bipartitions. The value of the saturation of the final state entanglement is much higher than the GHZ case. First, let us discuss how the entanglement of the output state ($E_N$) changes with the increasing number of qubits $N$ in the region of saturation (we choose, $t=30$). Two corresponding situations i.e. ``$1$-Rest" and ``highest-cut" are plotted in Fig (\ref{E_N_s_t30_W} (a)) and (\ref{E_N_s_t30_W} (b)) respectively. In the ``$1$-Rest" case, note that the decay of entanglement almost coincides for $s=2$ and $s=3$ while the best is obtained again between $s=2.3$ and $s=2.5$. Now for the ``highest-cut" scenario the entanglement of the output state shows a particularly interesting feature when the initial state is a W state. In this case, we see that for even number of qubits, the saturation occurs at the same value making the output state entanglement independent with the number of qubits. Whereas, this is not the case for odd number of qubits. Instead the saturation value increases with $N$ as can be seen from Fig (\ref{E_N_s_t30_W} (b)). Moreover we find, similar to the case of the GHZ states, the best effect of non-Markovianity is achieved for $s=2.47$.

Now let us concentrate on the particular case with $s=2.47$ and the time $t=30$. ``$1$-Rest" case decay of final state entanglement takes the similar form as GHZ state. We have, $E_N = a e^{-c (N-3)} +b^2$ for all values of $s>2$. For $s=2.47$, the parameters take the following values: $a=0.2329$, $b=0.4906$ and $c=0.1559$. For other values of $s$, the final state entanglement follows the similar structure with different parametrisation. 
This scenario is illustrated in Fig (\ref{E_N_s_t30_W} (c)) where the inset picture gives the numerically obtained situation for $N=3$ to $N=10$. By extrapolating with respect to the given function, we see that the final state entanglement in ``$1$-Rest" case saturates at a non-zero value, in particular $E_N=0.24$ as illustrated in the same figure making W states $24\%$ robust under dephasing noise in ``$1$-Rest" cut. Now let us move on to the ``highest-cut" case. The inset figure of Fig (\ref{E_N_s_t30_W} (d)) portrays the two distinct cases of final state entanglement for even and odd number of qubits separately. For even $N$, the final state entanglement saturates at a value $E_N=0.4977$ with $N$ and for the odd case, the output state entanglement variation with $N$ is governed by $E_N=1/(a e^{-cN}+b^2)$ with $a=0.2832$, $b=1.4195$ and $c=0.4546$. Extrapolating these data points, one can obtain the characterisation of entanglement for large $N$ as shown n Fig (\ref{E_N_s_t30_W} (d)). While for even number of qubits, the value of $E_N$ does not change, for odd number of qubits it increases and saturates at a value of $E_N=0.4962$. This evidently indicates that there exists a clear dichotomy between the even and odd number of qubits in case of multipartite system when we are initiating any task with W state. Note that, in the even qubit scenario the W states are $49.77\%$ robust whereas in the odd qubit scenario it $49.62\%$ robust in presence of dephasing noise.\\

From the detailed analysis of the GHZ and W state under dephasing evolution, it is clear that the W state is way more robust than GHZ state under this type of noise. While for GHZ state, entanglement vanishes with increasing number of parties, W state retains a non-zero amount of entanglement for any bipartition in the asymptotic limit. Also, while there is an information back-flow from the environment to the system (when $\gamma(t)$ takes negative value), evidently the system retains more amount of correlation under this noise model compared to the Markovian version of the same channel (i.e. for $s<2$). \\

\noindent \emph{Results for Depolarising noise.} As discussed in the previous section, in case of the  depolarising noise, we choose either of the $\gamma$'s such that they take negative values for certain interval of time of the evolution. First we consider the $\gamma_z(t)$ to be of the form given in Eq. (\ref{gamma_dephase}) as the corresponding part of the master equation in Eq. (\ref{depo_ME}) is same as the dephasing noise along with $\gamma_x(t)=\gamma_y(t)=0.1$. In this scenario, though the conditions given via Eq. (\ref{p-div_cond}) are fulfilled for the dynamics, it does not show any significant improvement of final state entanglement with time. 

Next we choose the case, $\gamma_z(t)= \alpha \sin(t)$ with $\alpha=1.0$ along with $\gamma_x(t)=\gamma_y(t)=0.1$. Let us consider the case when the initial state is GHZ. The corresponding scenario is depicted in Fig. (\ref{E_t_GHZ_depo}). The prescription of the interaction with the environment at a given instant of time is similar to the previous case, i.e. each qubit interact with identical reservoirs corresponding to the depolarising noise.But unlike dephasing scenario, here when the individual qubit undergoes interaction with the noise, the final state entanglement takes different values for different bipartitions even for GHZ state. Hence as before we again show the entanglement characterisation for two cases, namely, ``$1$-Rest" and ``highest-cut" in Fig (\ref{E_t_GHZ_depo} (a)) and (b) respectively. We deal explicitly with qubits from $N=3$ to $N=10$. Note that, in both the situations of bipartitions, the entanglement of the final state falls in the beginning and then revives a non-zero value after certain period of time after the collapse. In the ``highest-cut" scenario, for $N\geq7$ there is a second round of entanglement revival that occurs around $t=12$ to $t=13$ which is shown in the inset of Fig (\ref{E_t_GHZ_depo} (b)). The same study has been done for $N$-qubit systems when the initialisation is done in W state. Here, we portray the case of W state with either single or $(N-1)$ mode excitation. In the case of W states, the output state entanglement characterisation is plotted in Fig. (\ref{E_t_W_depo}). Note that, in case of depolarising noise GHZ states are more robust than W states with increasing number of parties, which is the opposite situation compared to the dephasing noise.

\begin{figure}[htp]
\centering
\fbox{
\subfigure[Entanglement characterisation in ``$1$-Rest" scenario.]{\includegraphics[scale=0.1]{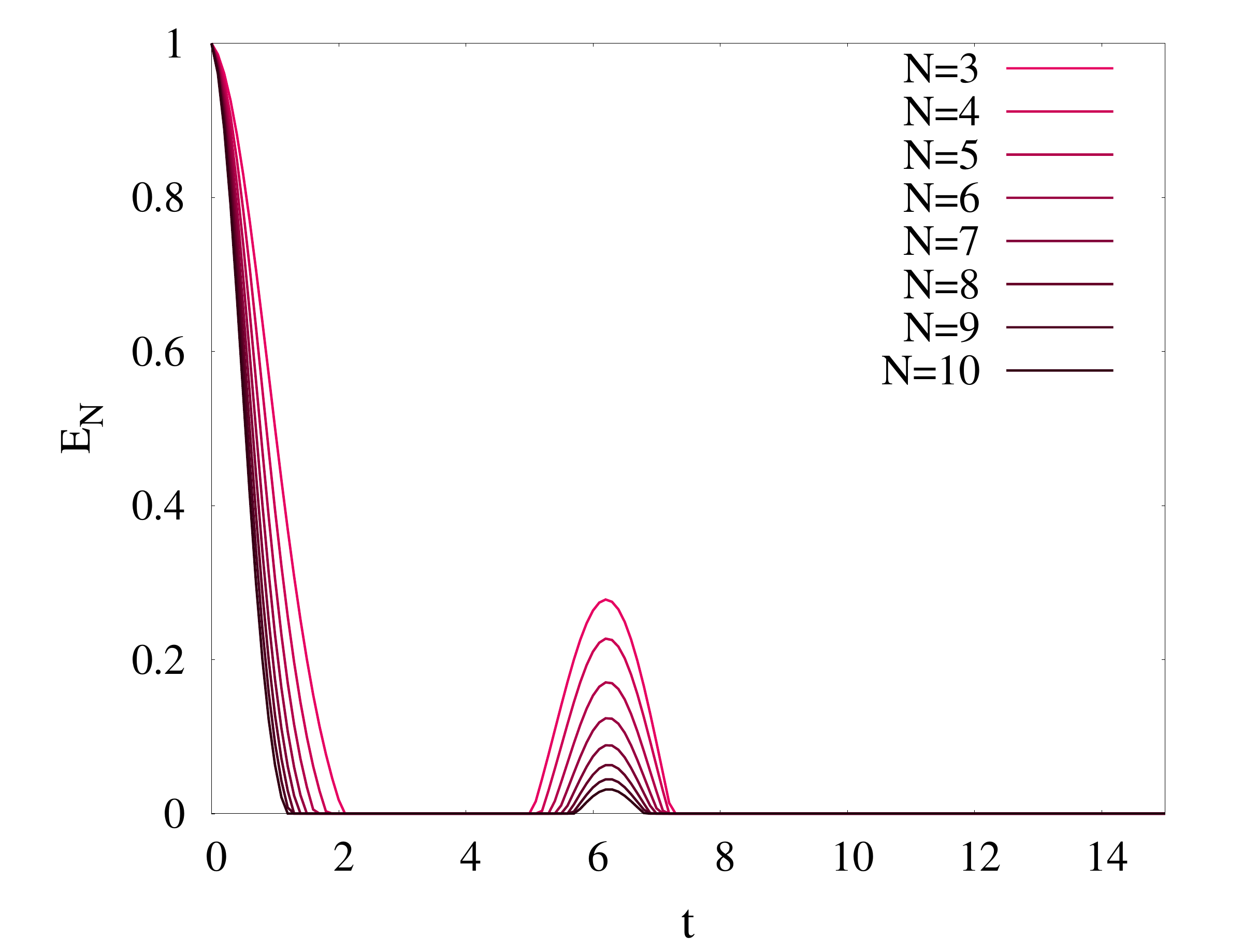}}
\quad
\subfigure[Entanglement characterisation in ``highest-cut" scenario.]{\includegraphics[scale=0.1]{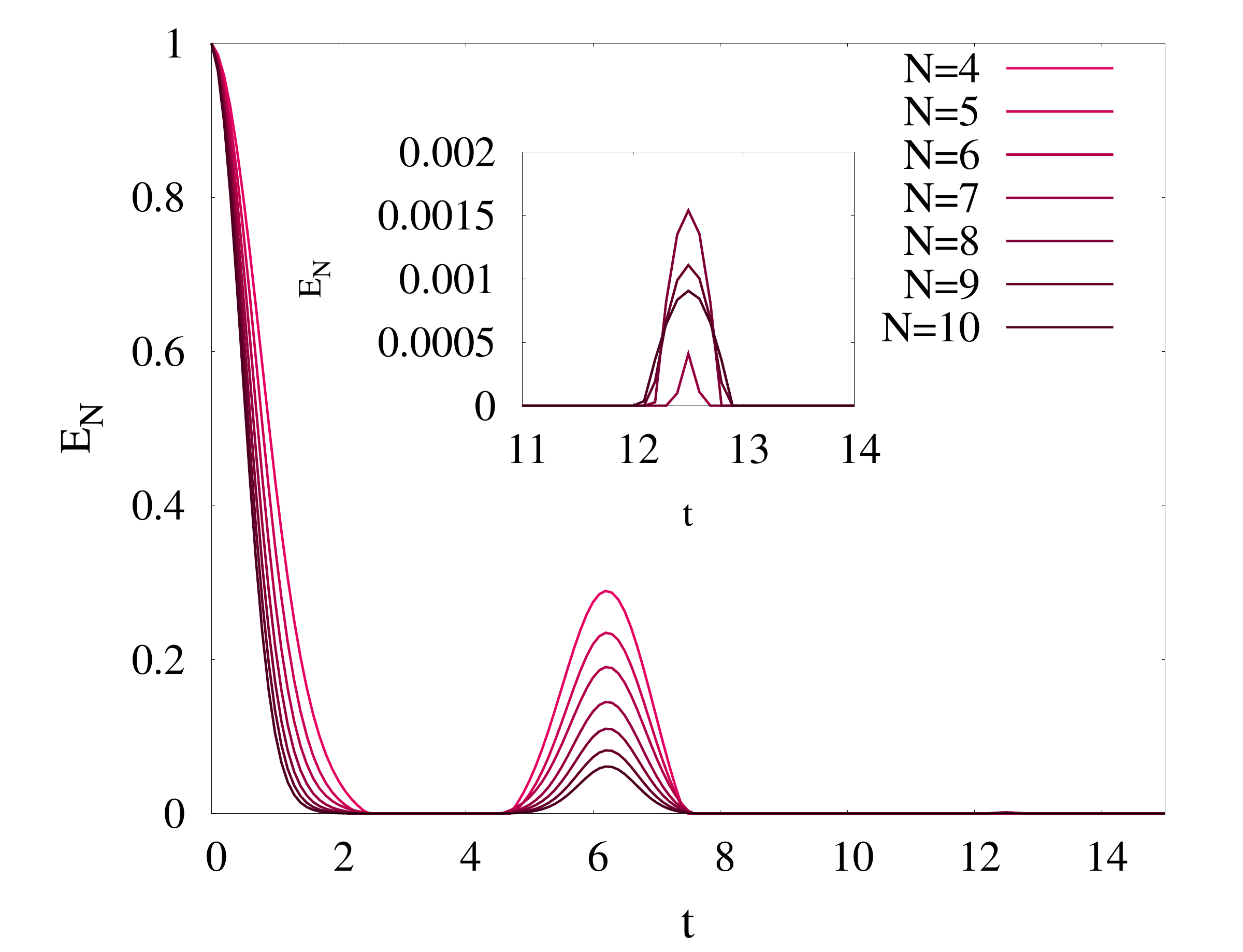}}}
\caption{\footnotesize{Entanglement of the final state under depolarising noise with time when the initial state is GHZ. The vertical axes are in ebits, while the horizontal ones are dimensionless. 
}}
\label{E_t_GHZ_depo}
\end{figure}

\begin{figure}[htp]
\centering
\fbox{
\subfigure[Entanglement characterisation in ``$1$-Rest" scenario.]{\includegraphics[scale=0.1]{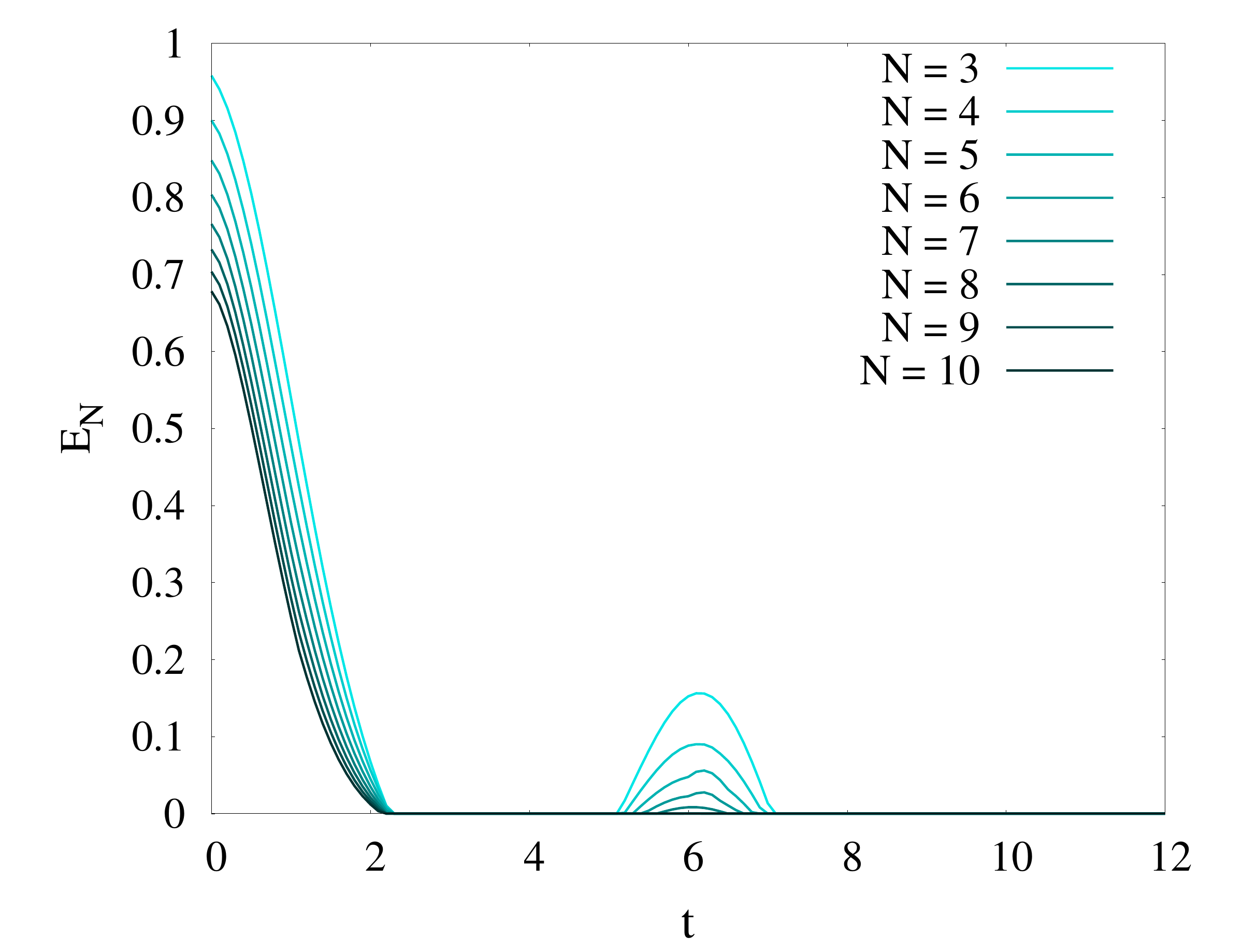}}
\quad
\subfigure[Entanglement characterisation in ``highest-cut" scenario.]{\includegraphics[scale=0.1]{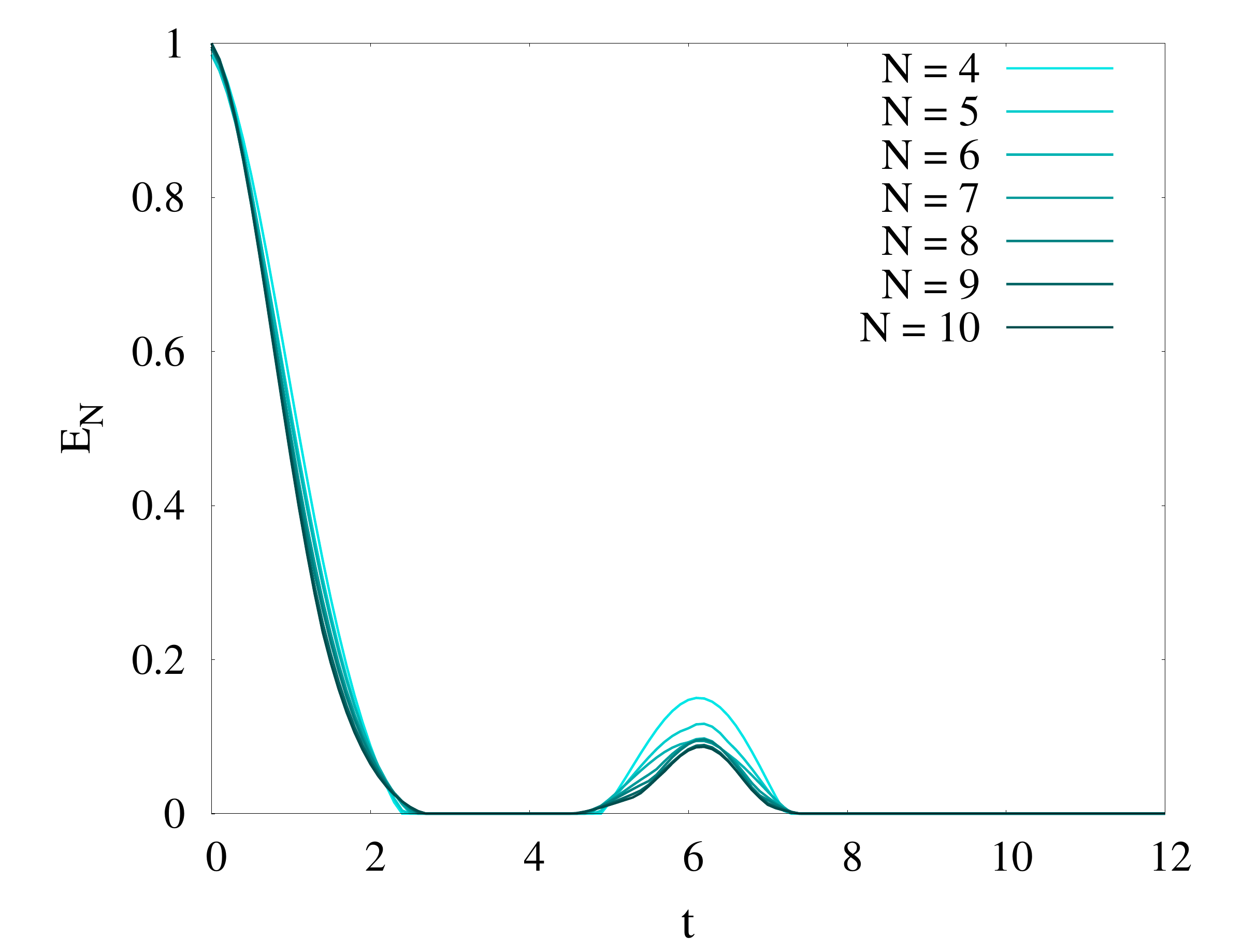}}}
\caption{\footnotesize{Entanglement of the final state under depolarising noise with time when the initial state is W. The vertical axes are in ebits, while the horizontal ones are dimensionless. }}
\label{E_t_W_depo}
\end{figure}

\section{Conclusion}
\label{sec_4}

In this article, we presented a detailed analysis regarding entanglement characterisation of certain multipartite quantum states when exposed to different noisy environments. We considered  the standard \(N\)-party Greenberger-Horne-Zeilinger and W  states as the multiparty states. The noise models taken into account were dephasing and depolarising ones. The noisy evolution of the system was considered to be governed by the Gorini-Kossakowski-Sudarshan-Lindblad master equation. By imposing non-Markovian versions of these channels, we obtained the evolved states at arbitrary instants of time. Comparison of the channels in their Markovian and non-Markovian domains made it  evident that the back-flow of information from the environment to the system, whenever present, induces an effect on the entanglement of the final states, making them more robust with respect to the Markovian case. 

The dynamical picture indicated that the W  states are more robust than GHZ in presence of dephasing noise. While we numerically obtained the exact data for $3$-qubit to $10$-qubit systems, we characterised the entanglement content of the final state for larger number of qubits by extrapolation. While GHZ states loses the retained amount of entanglement with  increasing number of qubits, we showed that  W  states become robust, with a non-zero value of entanglement, in presence of dephasing noise, and this is valid even in the asymptotic regime. We predicted the limit of  sustainable entanglement of the final states in both the cases. Alongside, we found that there exists an even-odd dichotomy for number of qubits in case of W states when exposed to dephasing noise. In case of an even number of qubits, the W states were shown to be $49.77\%$ robust, whereas in case of an odd number of qubits, the remaining entanglement saturated at a particular value, making the states $49.62\%$ robust under the particular noise, in the asymptotic limit. 

Next, in presence of depolarising noise, we found that when we consider the Born-Markov approximation, the dynamics governed by the Markovian master equation diminishes the presence of entanglement with time and the amount of the remaining entanglement vanishes 
with an increase in the number of qubits. Considering the corresponding non-Markovian evolution, we showed that it is possible to regain the lost entanglement with time, via a proper choice of decay parameters, after the collapse.

\section*{Acknowledgement}

SG would like to thank Samyadeb Bhattacharya and Srijon Ghosh for fruitful discussions. The authors acknowledge partial support from the Department of Science and Technology, Government of India through the QuEST grant (grant number DST/ICPS/QUST/Theme-3/2019/120).

\bibliography{ref_NM_Ent}
    
\end{document}